\documentclass[fleqn]{iopart}

\usepackage{iopams}  
\expandafter\let\csname equation*\endcsname\relax

\expandafter\let\csname endequation*\endcsname\relax

\usepackage{vmargin}
\usepackage{natbib}  
\usepackage{algorithmicx}
\usepackage{algorithm}
\usepackage{graphicx}

\usepackage{amsfonts}
\usepackage{stackrel}
\usepackage{amssymb}
\usepackage{mathrsfs}
\usepackage{algpascal}

\usepackage{amsmath}
\usepackage{amsthm}

\newtheorem{theorem}{Theorem}[section]

\newtheorem{remark}[theorem]{Remark}
\newtheorem{definition}[theorem]{Definition}

\usepackage{color}
\newcommand{\JMcomm}[1]{{\textcolor{red}{ #1}}}

\makeatletter
\def\footnoterule{\kern-3\p@
  \hrule \@width 2in \kern 2.6\p@} 
\makeatother

\begin{document}

\title[]{Notes on open questions within density functional theory
(existence of a derivative of the Lieb functional in a restricted sense and non-interacting $v$-representability)}

\author{
J{\'e}r{\'e}mie Messud
}


\vspace{10pt}
\begin{indented}
\item[]March 8, 2021
\end{indented}

\begin{abstract}
Density functional theory together with the Kohn-Sham scheme represent an efficient framework to recover the ground state density and energy of a many-body quantum system from an auxiliary ``non-interacting'' system (one-body with a local potential).
However, theoretical questions remain open.
An important one is related to the existence of a derivative of the Lieb functional in some restricted (weak) sense (as pointed out by Lammert).
Then, a further one would be related to the validity of the ``non-interacting $v$-representability'' conjecture.
We gather here elements on these questions,
providing reminders and thoughts that are non-conclusive but hopefully contribute to the reflection.
This document must be considered as notes (it does not represent an article).
\end{abstract}

%
%

%

%
%

\section{Introduction}

Density functional theory (DFT)~\cite[]{Hoh64,Par89,Dre90,Koh99r,Engel2011} has
become over the last decades a widely used theoretical tool for the
description and analysis of particles properties in quantum systems.
The essence of DFT is the Hohenberg-Kohn (HK) theorem~\cite[]{Hoh64},
ensuring that ground state many-particle systems in an ``external'' potential $v$ can be completely characterized by the ground state particle density~\cite[]{Par89,Dre90,Engel2011}.
Indeed, the HK theorem implies that any ground state observable of such a system can be written as a density functional, in particular the energy;
hence the possibility of substituting the involved (many-body and non-local) ``exchange-correlation'' contribution to the energy by a functional of the simple particle density. 
The form of the energy functional may be very complicated, but the proof of its existence leads to the appealing idea of replacing the original interacting (or many-body) system by a non-interacting (or one-body) system that exactly reproduces the ground state density and energy of the interacting system.
This reformulation has originally been proposed by~\cite{Koh65} (KS).
All quantum effects, including the exchange-correlation ones, are there described through a one-body local potential, 
called the KS potential~\cite[]{Koh65,Par89,Dre90,Engel2011}.
Consequently, the KS scheme may be computationally very efficient to describe systems with a sizeable number of particles.

Unfortunately, DFT states existence theorems but does not give a clue on the energy functional form.
A whole branch of DFT research is dedicated to set up approximate parameterized functionals.
These are usually based on the local density approximation (LDA), which is widely used to describe, for instance, many-electron systems in a molecule~\cite[]{Par89,Dre90,Engel2011}.

Another branch of DFT research is dedicated to bringing rigorous foundations to the formalism.
In particular, as the KS scheme implementation involves a derivative of the energy functional,
the existence of such a derivative must be established on solid ground.
This implies, among other things, the energy functional to be defined on a sufficiently large and dense set of admissible densities, the set being provided with a norm.
From that perspective, the most promising extension of the original energy functional proposal
seems the \cite{Lieb1983} functional.
However, even if \cite{Englisch1984,Englisch1984b} claimed to have demonstrated the existence of a derivative
of the Lieb functional,
\cite{Lammert2007} pointed out that the functional is in fact nowhere G\^ateaux differentiable.
\cite{Lammert2007} discussed how a derivative may exist in a weaker sense, considering 
among other restrictions on the perturbation (or density variation) directions
(which hopefully are the physical ones).
Taking another point of view, \cite{Kvaal2014} proposed to include a regularization term in the Lieb functional
and to work in a finite volume to obtain a differentiable DFT formulation.
However, even if the volume can be arbitrarilly large, the formalism cannot be applied to an infinite volume like in usual quantum mechanics frameworks.
The existence of a derivative of the Lieb functional in the general case thus still represents an open question.

Now, even if the Lieb functional could be proven admit a derivative in a restricted sense,
a further difficulty may occur:
the interacting and non-interacting Lieb functionals may not be differentiable
on the same set.
This difficulty 
has been called ``non-interacting $v$-representability question''~\cite[]{Par89,Dre90,Engel2011}.

The first part of these notes provides a review of the Lieb functional and of the KS scheme.
\JMcomm{
The ``false'' theorems claimed by \cite{Englisch1984,Englisch1984b} and their consequences are mentioned, as these are very common in DFT, but are colored in red to highlight they are false.
}
Then, the work of \cite{Lammert2007} and open questions are mentioned.
The second part provides non-conclusive thoughts on the non-interacting $v$-representability question.
Among others, static linear response considerations are applied to the non-interacting system
(introducing a specific equivalence class for degenerate ground state densities).
Our hope is that these notes contribute to the reflection, possibly for a more conclusive exploration.


\section{Lieb functional and Kohn-Sham scheme reminders}
\label{sec:reminders}

\subsection{{{Interacting system and Lieb functional}}}
\label{sec:Notations}

We consider a stationary quantum system of $N$ interacting identical particles, called interacting system,
whose Hamiltonian is
\begin{eqnarray}
\label{eq:H}
\hat{H}_v = \hat{T}+\hat{W}+\hat{V}[v]
.
\end{eqnarray}
$\hat{T}$ denotes the kinetic energy operator and $\hat{W}$ denotes the particles interaction operator (supposed to be symetric and spin independent).
$\hat{V}[v]=\int_{\mathbb{R}^3}v(r)\hat{n}(r)dr$ denotes a potential operator defined through an ``external'' potential $v(r)$, and the diagonal part of the particle density operator,
\begin{eqnarray}
\label{eq:n1}
\hat{n}(r) = \sum_{i=1}^{N}\delta(r-r_i)
.
\end{eqnarray}
$r_i$, indexed by $i\in\{1,...,N\}$, denotes the position of a particle.
Any following reasoning consider $(r,r_i)\in \mathbb{R}^3\times \mathbb{R}^3$, and fixed $N$ and $\hat{W}$.

A Banach space is a (possibly infinite dimensional) vector space with a norm
\cite[]{Rudin1991,Brezis1983}.
We consider the $\mathscr{L}^{p}$ ($p\ge 1$) Banach spaces defined by 
%
\begin{eqnarray}
\label{eq:Banach}
\mathscr{L}^p(\mathbb{R}^M)
=
\Big\{
f(x) \Big|
x\in\mathbb{R}^M,
||f||_p<\infty
\Big\}
\text{ with norm }
||f||_p=\Big(\int_{\mathbb{R}^M}|f(x)|dx\Big)^\frac{1}{p}
.
\end{eqnarray}
When $p=\infty$, $||f||_\infty=\sup_{r\in\mathbb{R}^3}|f(r)|$ which is called the uniform norm.
The Banach space of admissible potentials $v$ for the Hamiltonian in eq.~(\ref{eq:H}) is
\cite[]{Lieb1983,vanLeeuwen2003,Engel2011}
\begin{eqnarray}
\label{eq:V}
&&
\mathscr{V}
=
\Big\{
v=v_1+v_2 \Big|
v_1\in\mathscr{L}^\frac{3}{2}(\mathbb{R}^3),
v_2\in\mathscr{L}^\infty(\mathbb{R}^3)
\Big\}
\\
&&
\hspace{2cm}
\text{ with norm }
||v||_{\mathscr{V}}=
\stackrel[v_1+v_2=v]{}{\inf_{v_1\in\mathscr{L}^\frac{3}{2}(\mathbb{R}^3),v_2\in\mathscr{L}^\infty(\mathbb{R}^3)}}
||v_1||_\frac{3}{2}+||v_2||_\infty
,
\nonumber
\end{eqnarray}
where potentials that differ by an additive constant are identified.

We also consider the Sobolev space
\begin{eqnarray}
\label{eq:Sobolev}
&&
\mathscr{H}^1(\mathbb{R}^M)
=
\Big\{
f \Big|
f\in\mathscr{L}^2(\mathbb{R}^M),
\nabla f\in\vec{\mathscr{L}^2}(\mathbb{R}^M)
\Big\}
\\
&&
\hspace{2cm}
\text{ with norm }
||f||_\mathscr{H}=\Big(\int_{\mathbb{R}^M}(|f(x)|^2+|\nabla f(x)|^2)dx\Big)^\frac{1}{2}
.
\nonumber
\end{eqnarray}
In the following, we will be interested in the particle density defined through eq.~(\ref{eq:n1}) by
\begin{eqnarray}
\label{eq:n2}
n(r) = \langle\psi|\hat{n}(r)|\psi\rangle
,
\end{eqnarray}
where $|\psi\rangle$ can represent any eigenstate of $\hat{H}_v$,
thus among others the ground state that will interest us further.
More generally, $|\psi\rangle$ can be any state with $\psi\in \mathscr{H}^1(\mathbb{R}^{3N})$.
A general convex set of pertinent densities is \cite[]{Lieb1983,vanLeeuwen2003,Engel2011}
\begin{eqnarray}
\label{eq:S}
\mathscr{S}
=
\Big\{
n \Big|
n(r)\ge 0, \int_{\mathbb{R}^3}n(r)dr=N,\sqrt{n}\in \mathscr{H}^1(\mathbb{R}^3)
\Big\}
,
\end{eqnarray}
called ``$N$-representable densities set''.
The $||.||_{p=1}$ and $||.||_{p=3}$ norms, here denoted by $||.||_{1,3}$, can be associated to $\mathscr{S}$ as $\mathscr{S}\subset\mathscr{X}$
where $\mathscr{X}$ is the Banach space\footnote[1]{
The topological dual $\mathscr{X}^*$ of $\mathscr{X}$ is ``represented'' by $\mathscr{V}$, eq.~(\ref{eq:V}).
Similarly, $||.||_{\mathscr{V}}$ ``represents'' the dual norm of $||.||_{1,3}$.
See Riesz representation theorem related considerations for instance in \cite{Brezis1983} (chapter I) and \cite{Rudin1991} (chapter IV).
}
\begin{eqnarray}
\label{eq:N}
\mathscr{X}
=
\mathscr{L}^1(\mathbb{R}^3)
\cap
\mathscr{L}^3(\mathbb{R}^3)
.
\end{eqnarray}

\cite{Lieb1983} proposed a ``universal'' functional of the density
(for given $N$ and $\hat{W}$), here called the interacting Lieb functional, defined by
\begin{eqnarray}
\label{eq:Lieb}
\forall n\in\mathscr{S}:
\quad
&F_L[n]=&\inf_{\hat{D}\in\mathscr{D}(n)}
\text{Tr } \hat{D}(\hat{T}+\hat{W})
,
\end{eqnarray}
where Tr denotes the trace operator
and $\mathscr{D}(n)$ denotes the set of $N$-particle density matrices that yield a gives density $n$
\begin{eqnarray}
\label{eq:D}
\mathscr{D}(n)
=
\Big\{
\hat{D} \Big|
&&
\hat{D}=\sum_k d_k |
\psi_k\rangle\langle\psi_k|,d_k^*=d_k\ge 0,\sum_k d_k=1,
\\
&&
\langle\psi_k|\psi_l\rangle=\delta_{kl},\psi_k\in \mathscr{H}^1(\mathbb{R}^{3N}),
n(r)=\sum_k d_k \langle\psi_k|\hat{n}(r)|\psi_k\rangle
\Big\}
.
\nonumber\
\end{eqnarray}
The interacting Lieb functional can be extended to the whole Banach space $\mathscr{X}$ by posing, in addition to eq.~(\ref{eq:Lieb}),
\begin{eqnarray}
\label{eq:Lieb2}
\forall n\in \mathscr{X}\backslash\mathscr{S}:
&F_L[n]=& +\infty
.
\end{eqnarray}
$\mathscr{S}$ represents the domain of $F_L[n]$.
\begin{theorem}[Demonstrated by \cite{Lieb1983}]
\label{th:Lieb_csc}
$F_L[n]$ is convex and lower semi-continuous on $\mathscr{X}$.
\end{theorem}
\noindent
The ``interacting energy functional'' is defined by
\begin{eqnarray}
\label{eq:Lieb_en}
\forall n\in\mathscr{S}:
\quad
E_v[n] = F_L[n]+\int_{\mathbb{R}^3}v(r)n(r)dr
.
\end{eqnarray}
%
$E_v[n]$ is convex with respect to $n$ as $F_L[n]$ is.

\begin{theorem}[Demonstrated by \cite{Lieb1983}]
The minimum of the interacting energy functional,
\begin{eqnarray}
\label{eq:Lieb_en_v}
E[v] = \inf_{n\in \mathscr{S}} E_v[n]
,
\end{eqnarray}
equals the ground state energy of the Hamiltonian in eq.~(\ref{eq:H})
for any external potential $v\in\mathscr{V}$ that yields a ground state.
\end{theorem}

In the non-degenerate case, i.e. when a unique ground state is associated to $\hat{H}_v$, there is a one-to-one correspondence between $v$ and the corresponding minimizing (ground state) density.
In the degenerate case, however, the many minimizing (ground states) densities related to the degenerate ground states of $\hat{H}_v$ are identified;
this provides a one-to-one correspondence between $v$ and the equivalence {class} of these ground state densities~\cite[]{Par89,Dre90,Engel2011}.
Then, the potential $v$ can be considered as a functional of the ground state density and conversely.

\JMcomm{
\begin{theorem}[False theorem by \cite{Englisch1984,Englisch1984b}]
\label{th:Lieb_diff}
$F_L[n]$ admits a G\^ateaux derivative only in a subset $\mathscr{B}\subset\mathscr{S}$, where $\mathscr{B}$ is the set of interacting (ensemble) $v$-representable densities defined by
\begin{eqnarray}
\label{eq:B}
\mathscr{B}
=
\Big\{
n \Big|
&&
n(r)=\sum_{i=1}^q c_i \langle\psi_i|\hat{n}(r)|\psi_i\rangle,c_i^*=c_i\ge 0,\sum_{i=1}^q c_i=1,
\\
&&
|\psi_i\rangle \text{ belongs to the set of ground states of } \hat{H}_v \text{ for some } v
\Big\}
.
\nonumber
\end{eqnarray}
The derivative is defined by
\begin{eqnarray}
\label{eq:Lieb_diff}
\forall n\in\mathscr{B}:
\quad
\frac{\partial F_L[n]}{\partial n(r)}=-v(r)
,
\end{eqnarray}
where $v\in \mathscr{V}$ would be the potential that generates $n\in\mathscr{B}$,
considering the equivalence class of potentials defined up to an additive constant.
\end{theorem}
}

\noindent
The set $\mathscr{B}$ includes the interacting pure state $v$-representable densities ($q=1$ restriction).
The $q>1$ possibility allows to deal with the general degenerate (ensemble) interacting system case.

\subsection{{Non-interacting system}}
\label{sec:Non-int}

We now consider a stationary system of $N$ identical particles whose Hamiltonian is
\begin{eqnarray}
\label{eq:H_s}
\hat{H}_s = \hat{T}+\hat{V}[v_s]
,
\end{eqnarray}
where $\hat{V}_s[v_s]=\int_{\mathbb{R}^3}v_s(r)\hat{n}(r)dr$ denotes a potential operator
defined through a potential $v_s(r)$.
Compared to the interacting system Hamiltonian in eq.~(\ref{eq:H}),
the non-interacting system Hamiltonian in eq.~(\ref{eq:H_s}) omits the particles interaction operator.
%
The ``non-interacting Lieb functional'' is defined through
\begin{eqnarray}
\label{eq:Lieb_s}
\forall n\in\mathscr{S}:
\quad
T_L[n]=\inf_{\hat{D}\in\mathscr{D}(n)}
\text{Tr } \hat{D}\hat{T}
,
\end{eqnarray}
where the particles interaction term is omitted compared to the interacting Lieb functional in eq.~(\ref{eq:Lieb}).

As in \S\ref{sec:Notations}, 
there is in the non-degenerate case a one-to-one correspondence between $v_s$ and the corresponding minimizing density $n_s$.
In the degenerate case, there is a one-to-one correspondence between $v_s$ and the equivalence {class} of ground state densities obtained from the degenerate ground states of $\hat{H}_s$.
Then, the potential $v_s$ can be considered as a functional of the density $n_s$, i.e. $v_s[n_s]$, and conversely~\cite[]{Par89,Dre90,Engel2011}.
%

\JMcomm{
A consequence of the \textbf{False Theorem \ref{th:Lieb_diff}} would be:
$T_L[n]$ admits a G\^ateaux derivative
only in a subset $\mathscr{B}_s\subset\mathscr{S}$,
defined by~\cite[]{Englisch1984,Englisch1984b}
\begin{eqnarray}
\label{eq:Lieb_diff_s}
\forall n_s\in\mathscr{B}_s:
\quad
\frac{\partial T_L[n_s]}{\partial n_s(r)}=-v_s(r)
,
\end{eqnarray}
where $v_s\in \mathscr{V}$ is the potential that generates $n_s\in\mathscr{B}_s$,
called the KS potential, considering the equivalence class of potentials defined up to an additive constant.
$\mathscr{B}_s$ is the set of non-interacting (ensemble) $v$-representable densities defined by
\begin{eqnarray}
\label{eq:B_s}
\mathscr{B}_s
=
\Big\{
n_s \Big|
&&
n_s(r)=\sum_{i=1}^{q_s} c_{s_i} \langle\psi_{s_i}|\hat{n}(r)|\psi_{s_i}\rangle,c_{s_i}^*=c_{s_i}\ge 0,\sum_{i=1}^{q_s} c_{s_i}=1,
\\
&&
|\psi_{s_i}\rangle \text{ belongs to the set of ground states of } \hat{H}_s  \text{ for some } v_s
\Big\}
.
\nonumber
\end{eqnarray}
}

\noindent
The set $\mathscr{B}_s$ includes the non-interacting pure state $v$-representable densities set denoted by $\mathscr{A}_s$ ($q_s=1$ restriction).
The $q_s>1$ possibility allows to deal with the general degenerate non-interacting system case.

\JMcomm{
We now conjecture that some (possibly weaker but sufficient) form of \textbf{Theorem \ref{th:Lieb_diff}} exists
and follow the traditional DFT reasonment~\cite[]{Par89,Dre90,Engel2011} as a reminder.
From eqs.~(\ref{eq:Lieb_diff}) and (\ref{eq:Lieb_diff_s}), we notice that the sets on which the interacting and non-interacting Lieb functionals would be differentiable,
respectively $\mathscr{B}$ and $\mathscr{B}_s$, are not the same.
Let us first consider the case of an interacting system ground state density $n\in \mathscr{B}\cap\mathscr{B}_s$.
The main outcome of previous considerations is that such densities can be reproduced by a non-interacting system using in eq. (\ref{eq:H_s}) the potential $v_s[n]$ computed by eq.~(\ref{eq:Lieb_diff_s}) with $n_s=n$,
i.e. through
\begin{eqnarray}
\label{eq:KS}
&&
\forall k\in\{1,...,q_s\}:\quad
\Big(\hat{T}+\hat{V}[v_s[n]]\Big) |\psi_{s_k}\rangle
=
E_{s_1}|\psi_{s_k}\rangle
\\
&&
n(r)=\sum_{k=1}^{q_s} \lambda_k n_k(r),
\quad
n_k(r)=\langle\psi_{s_k}|\hat{n}(r)|\psi_{s_k}\rangle,
\quad
\lambda_k^*=\lambda_k\ge 0,
\sum_{k=1}^{q_s} \lambda_k=1
,
\nonumber
\end{eqnarray}
where $q_s$ degenerate ground states $|\psi_{s_k}\rangle$ with energy $E_{s_1}$ are possibly considered for the non-interacting system when the computed potential $v_s[n]$ leads to such a situation
(we have $q_s=1$ when $v_s[n]$ does not lead to ground state degeneracy).
}
The $\{|\psi_{s_k}\rangle\}$ can be chosen to be orthonormal eigenstates, for instance Slater determinants in the case of a fermionic system~\cite[]{Dre90,Engel2011}.
When the non-interacting system is degenerate,
at least one of the many minimizing densities $\sum_{k=1}^{q_s} \lambda_k \langle\psi_{s_i}|\hat{n}(r)|\psi_{s_i}\rangle$ agrees with the interacting system ground state density, which fixes the $\lambda_k$.
The equivalence class of these degenerate densities is thus identified to this element.

\JMcomm{
Now, what about an interacting system ground state density $n\in\mathscr{B}\backslash\mathscr{B}_s$?
We would have the following theorem (again conditioned by the existence of some form of \textbf{Theorem \ref{th:Lieb_diff}}):
%
\begin{theorem}[False theorem by \cite{Englisch1984,Englisch1984b}]
\label{th:Lieb_dense}
\begin{eqnarray}
\label{eq:epsilon_1}
&&
\forall n\in\mathscr{B}, \forall \epsilon>0,
\\
&&\hspace{1.5cm}
\exists n_\epsilon\in\mathscr{B}_s:
\quad
||n_\epsilon-n||_{1,3}\le\epsilon
.
\nonumber\\
&&
\text{If a derivative of $T_L$ exists, then }
n_\epsilon
\text{ can represent a ground state density of the }
\nonumber\\
&&
\text{non-interacting system with potential }
v_s[n_\epsilon]=-\frac{\partial T_L[n_\epsilon]}{\partial n_\epsilon(r)}
.
\nonumber
\end{eqnarray}
\end{theorem}
\noindent
Considering the functional form of $v_s[n]$ as known for $n\in\mathscr{B}\cap\mathscr{B}_s$,
we could self-consistently resolve eq.~(\ref{eq:KS}) to recover the ground state density of the corresponding interacting system.
If $\mathscr{B}=\mathscr{B}_s$, 
i.e. $\mathscr{B}\cap\mathscr{B}_s=\mathscr{B}$,
then a non-interacting system would theoretically always exist to reproduce exactly the ground state density of any interacting system~\cite[]{Par89,Dre90,Engel2011}.
}

\subsection{{Kohn-Sham scheme}}
\label{sec:KS}

A first step is usually done by decomposing $E_v[n]$ as~\cite[]{Koh65,Par89,Dre90,Engel2011}
\begin{eqnarray}
\label{eq:Lieb_en_prac}
\forall n\in\mathscr{S}:
\quad
&&
E_v[n] = T_L[n]+E_H[n] + E_{XC}[n] + \int_{\mathbb{R}^3}v(r)n(r)dr
,
\end{eqnarray}
where $E_H[n]$ is the Hartree energy and $E_{XC}[n]$ is the exchange-correlation energy defined by
\begin{eqnarray}
\forall n\in\mathscr{S}:
\quad
&&
E_{XC}[n]=F_L[n]-T_L[n]-E_H[n]
\nonumber
.
\end{eqnarray}
This rewriting of $E_v[n]$ makes sense as the unknown and ``complicated'' term, $E_{XC}[n]$, often represents a ``correction'', i.e. is quite small compared to the other terms in eq.~(\ref{eq:Lieb_en_prac}), which is pratically advantageous in the approximate parameterizations quest~\cite[]{Par89,Dre90,Engel2011}.
Note however that the functional form of $E_{XC}[n]$ is not given by DFT theorems, DFT only stating existence theorems.
Approximate parameterizations are thus used in practice, like the LDA one~\cite[]{Par89,Dre90,Engel2011}.

\JMcomm{
Then, the KS potential can be splitted into
\begin{eqnarray}
\label{eq:v_s_prac}
\forall n\in\mathscr{B}\cap\mathscr{B}_s:
\quad
v_s[n]
=
v_H[n]
+v_{XC}[n]
+v
,
\end{eqnarray}
where $v_H[n]=\partial E_{H}[n]/\partial n(r)$ is the Hartree potential, $v$ the external potential (the same than the one of the interacting system) and
$v_{XC}[n]$ the ``exchange-correlation'' potential defined by
%
%
\begin{eqnarray}
\label{eq:v_xc_dif}
v_{XC}[n](r)=\frac{\partial E_{XC}[n]}{\partial n(r)}
.
\end{eqnarray}
}

\section{Existence of a derivative of the Lieb functional?}
\label{sec:beyonf_subdif}

\begin{definition}[Directional derivative and derivative]
\label{def:diff}
We consider a Banach space denoted by $B$ and a functional $G[n]:B\rightarrow \mathbb{R}$.
The directional derivative evaluated at some $n\in B$ in the direction $h\in B$ is defined by
\begin{eqnarray}
\label{eq:diff_def_2}
G'[n,h]
=
\lim_{\epsilon\rightarrow 0^+}\frac{G[n+\epsilon h]-G[n]}{\epsilon}
,
\end{eqnarray}
when the limit exists, $\infty$  being accepted.
For a G\^ateaux derivative $\partial G[n]/\partial n(r)$  to exist, $G'[n,h]$ must be finite and linear in $h$, for all $h$ in $B$. This implies
%
\begin{eqnarray}
\label{eq:diff_def}
\exists \varphi[n]\in  B^*:
\quad
G'[n,h]=\int_{\mathbb{R}^3} \varphi[n](r) h(r)
\quad\Rightarrow\quad
\frac{\partial G[n]}{\partial n(r)}=\varphi[n](r)
.
\end{eqnarray}
%
If these properties do not hold for all $h$ in $B$ but only in some clearly defined and pertinent subset of $B$, a weaker derivative notion can be obtained.
\end{definition}


A main foundation of the KS scheme had for long been \textbf{Theorem \ref{th:Lieb_diff}}.
However, \cite{Lammert2007} pointed out that the Lieb functional is in fact nowhere G\^ateaux differentiable,
invalidating this Theorem.
Nevertheless, would $F_L[n]$ admit some weaker notion of a derivative in the subset $\mathscr{B}\subset\mathscr{S}$?
This would be sufficient for most of previous KS scheme considerations to hold.

So, let us consider $n\in\mathscr{B}$ in the following.
A necessary condition for a derivative to exist is $F_L'[n,h]$ to be finite and linear in $h$.
Because these properties cannot hold $\forall h\in\mathscr{X}$
among other because of eq.~(\ref{eq:Lieb2}),
$F_L[n]$ cannot admit a G\^ateaux derivative \cite[]{Lammert2007}.
However, a weaker notion of a derivative may still be defined if $F_L'[n,h]$ is finite and linear in $h$ for all $h$ in some clearly defined and (physically) pertinent subset of $\mathscr{X}$.
The ``perturbation directions'' $h$ must satisfy some regularity and particules number conservation constraints in order to be admissible;
indeed, $F_L'[n,h]$ can possibly be finite only for $n+\epsilon h\in\mathscr{S}$, remind eq.~(\ref{eq:Lieb2}).
Introducing the Banach space\footnote[1]{
$\delta\mathscr{S}$ has a vector space structure
unlike $\mathscr{S}$, eq.~(\ref{eq:S}), as linear combinations of elements in $\delta\mathscr{S}$ still integrate to $0$.
}
\begin{eqnarray}
\label{eq:S_0}
\delta\mathscr{S}
=
\Big\{
\delta n \Big|
\int_{\mathbb{R}^3}\delta n(r)dr=0,\sqrt{|\delta n|}\in \mathscr{H}^1(\mathbb{R}^3)
\Big\}
\quad
\text{ with norm $||.||_{1,3}$}
,
\end{eqnarray}
we may think to constrain $h$ to belong to $\delta\mathscr{S}$.
However, even if sufficient to conserve the particle number, this space
would not constrain $n+\epsilon h\ge 0$ even for infinitesimal $\epsilon$ on (non-null measure) sets where $n$ equals to zero.
$h\in\delta\mathscr{S}$ is thus not constraining enough to ensure $n+\epsilon h\in\mathscr{S}$ even in the limit $\epsilon\rightarrow 0^+$, unless if $n$ was restricted to some subset of $\mathscr{B}$ of strictly positive densities;
but this would be too restrictive physically.
To overcome that, a possibility could be to pose $h=n_1-n$ and search for some clearly defined and pertinent set for $n_1$ \cite[]{Lammert2007}, instead of searching directly a set for $h$.
A natural set for $n_1$ would be a subset of $\mathscr{S}$. Indeed, as $\mathscr{S}$ is convex,
$n+\epsilon h=n+\epsilon (n_1-n)=\epsilon n_1+(1-\epsilon)n\in\mathscr{S}$ for any $\epsilon\in[0,1]$, thus for $\epsilon\rightarrow 0^+$\footnote{
$h=n_1-n$ should belong to a subset of $\delta\mathscr{S}$ for all $n\in\mathscr{B}\subset\mathscr{S}$ and $n_1\in\mathscr{S}$.
}.
%

Then, for a derivative of $F_L[n]$ to possibly exist in a restricted sense,
a linearity property of $F_L'[n,n_1-n]$ with respect to $n_1-n$ must be demonstrated to hold for $n$ and $n_1$ in some clearly defined and pertinent subsets of $\mathscr{B}$ and $\mathscr{S}$, respectively.
\cite{Lammert2007} demonstrated such a property
for quite stongly constrained subsets (regularity and strict positivity constraints), that are too restrictive physically.
A more general demonstration that at least relaxes all constraints on $\mathscr{B}$ would be 
necessary for most of previous KS scheme considerations to hold.
This still represents an open question.

\JMcomm{
In the following, we conjecture that a derivative of $F_L[n]$ can be defined in some restricted (or weak) sense for all $n\in\mathscr{B}$, that allows to consider that previous KS scheme considerations hold and to discuss the non-interacting $v$-representability question.
Keep however in mind that the following relies on a conjecture that is still to be fully demonstrated.
}

\section{On non-interacting $v$-representability}

\subsection{{Non-interacting $v$-representability conjecture}}
\label{sec:Open}

As a consequence of conjecturing that a derivative of $F_L[n]$ can be defined in some restricted (or weak) sense for all $n\in\mathscr{B}$,
we consider that \textbf{Theorems \ref{th:Lieb_dense}} and \textbf{\ref{th:Lieb_dense}} are true when considering the weaker derivative form
and thus that a KS potential can be defined.
From eq.~(\ref{eq:epsilon_1}), we we can consider that the set $\mathscr{B}_s$, eq.~(\ref{eq:B_s}), is dense in the set $\mathscr{B}$, eq.~(\ref{eq:B}) (and conversely), and thus
\begin{eqnarray}
\label{eq:dens_lim}
\lim_{\epsilon\rightarrow 0} n_\epsilon=n
.
\end{eqnarray}
Two possibilities occur regarding this limit:
\begin{itemize}
\item
If $\mathscr{B}\ne\mathscr{B}_s$, the limit in eq.~(\ref{eq:dens_lim}) cannot be reached by $n_\epsilon$ when $n\in\mathscr{B}\backslash\mathscr{B}_s$.
A non-interacting system that reproduces $n$ with an arbitrary precision can still theoretically be defined,
but this does not mean there exists a practically friendly way to achieve it.
Indeed, $v_s[n_\epsilon]$ could then possibly change very rapidly (oscillate or even diverge)
to make $n_\epsilon$ slightly closer to $n$.
\item
If $\mathscr{B}=\mathscr{B}_s$, the limit is always reached, and conversely.
\end{itemize}
There is (even conjecturing that a KS potential can rigorously be defined from the Lieb functional) no proof that the sets $\mathscr{B}$ and $\mathscr{B}_s$ are equal.
Considering $\mathscr{B}=\mathscr{B}_s$, important for the complete foundation of the Lieb functional-based KS scheme, thus represents a further conjecture called ``non-interacting $v$-representability conjecture''.


\subsection{{Necessary condition}}
\label{sec:Non-int-conject}


%
A practical way to demonstrate the non-interacting $v$-representability conjecture could be:

\begin{remark}
[Non-interacting $v$-representability condition]
\label{th:non-int}
The non-interacting $v$-representability conjecture is true if
\begin{eqnarray}
\label{eq:epsilon_2}
&&
\forall n\in\mathscr{B}, \exists v_{n}\in \mathscr{V},
\\
&&\hspace{1.5cm}
\forall \epsilon>0,\exists \alpha_n^{(\epsilon)}\text{ satisfying } \lim_{\varepsilon\rightarrow 0} \alpha_n^{(\varepsilon)}=0,
\nonumber\\
&&\hspace{3cm}
\forall n_\epsilon\in\mathscr{B}_{s}^{(\epsilon,n)}:
\quad
\big|v_s[n_\epsilon](r)-v_{n}(r)\big|
\le \alpha_n^{(\epsilon)}(r)
,
\nonumber
\end{eqnarray}
where we considered the following set (defined by given $n\in\mathscr{B}$ and $\epsilon$ value)
\begin{eqnarray}
\label{eq:Bs_eps_n}
\mathscr{B}_{s}^{(\epsilon,n)}
=
\Big\{
n_\epsilon\in\mathscr{B}_s: ||n_\epsilon-n||_{1,3}\le\epsilon
\Big\}
.
\end{eqnarray}
%
$\mathscr{B}_{s}^{(\epsilon,n)}$ represents the subset of densities in $\mathscr{B}_{s}$
that are ``$\epsilon$-close'' to a given density $n\in\mathscr{B}$ from the $||.||_{1,3}$ norm point of view\footnote[1]{
Using this set, eq. (\ref{eq:epsilon_1}) can be reformulated:
$
\forall n\in\mathscr{B}, \forall \epsilon>0:
\mathscr{B}_{s}^{(\epsilon,n)}\ne\varnothing
$.
}.
Eq.~(\ref{eq:epsilon_2}) would imply
\begin{eqnarray}
\label{eq:lin_vn}
\lim_{\epsilon\rightarrow 0} v_s[n_\epsilon]=v_{n}
.
\end{eqnarray}
\end{remark}

\begin{proof}[\textbf{Justification}]
If the non-interacting $v$-representability conjecture is true, then any $n\in\mathscr{B}$ would be exactly reproductible by a non-interacting system
with a potential $v_s[n]$.
This is equivalent to say, using the notations in Theorem \ref{th:Lieb_dense}, that: (A) $v_s[n_\epsilon]$ should converge for $\epsilon\rightarrow 0$ towards an element $v_n\in\mathscr{V}$ in some smooth way
(i.e. that $v_s[n_\epsilon]$ should not change rapidly, oscillate... approaching the limit)~(\cite{vanLeeuwen2003}, section 15)
and (B) $v_s[n_\epsilon]$ should reach this $v_n$ (and not only tend to it without reaching it).
Satisfying eq.~(\ref{eq:epsilon_2}) would explicitly ensure point (A) and implicitly ensure point (B).
Indeed, as $v_s[n_\epsilon]$ and $v_{n}$ belong to the same space $\mathscr{V}$,
it would always be possible to infinitesimally ``tweak'' a $v_{n}$ that satisfies eq.~(\ref{eq:epsilon_2})
so that the limit is reached in eq.~(\ref{eq:lin_vn}).
Thus, proving eq.~(\ref{eq:epsilon_2}) would not only imply that 
$v_s[n_\epsilon]$ would converge towards some element in $\mathscr{V}$ in a smooth way,
but also that it is always possible to select this element so that the limit is reached.
Equivalently, the density domain limit in eq. (\ref{eq:dens_lim}) would always be reached,
implying through eq.~(\ref{eq:epsilon_1}) that $\mathscr{B}=\mathscr{B}_s$.
\end{proof}

\subsection{{Relationship with the existence of a bounded derivative for $v_s[n_s]$}}
\label{sec:Non-int-conject}

Demonstrating the existence of a bounded derivative of $v_s[n_s]$ for all $n_s\in\mathscr{B}_s$ would be sufficient to demonstrate that eq.~(\ref{eq:epsilon_2}) is true.
This makes sense physically: if in eq.~(\ref{eq:epsilon_2}) $v_s[n_\epsilon]$ varies sufficiently smoothly with respect to $n_\epsilon\in\mathscr{B}_s$ when $n_\epsilon$ approaches some limit $n\in\mathscr{B}$,
it should be possible at some point to continuously extend $v_s[n_\epsilon]$ at $n\in\mathscr{B}$. 
This should lead to the definition of an element $v_n\in\mathscr{V}$ towards which $v_s[n_\epsilon]$ converges smoothly for a $n_\epsilon$ sufficiently close to $n$ from the $||.||_{1,3}$ norm point of view.

\begin{theorem}[Non-interacting $v$-representability and derivative of the KS potential]
\label{th:1}
If a derivative of $v_s[n_s]$ exists in a restricted sense for $n_s\in\mathscr{B}_s$
and if this derivative is bounded,
then the non-interacting $v$-representability conjecture is true.
\end{theorem}

\begin{proof}[\textbf{Proof}]
\textbf{Theorem \ref{th:Lieb_dense}} ensures that $\mathscr{B}_s$ is dense in $\mathscr{B}$.
We moreover do the hypothesis a derivative of $v_s[n_s]$ exists in a restricted sense for all $n_s\in\mathscr{B}_s$
and that this derivative is bounded.
Mathematical analysis theorems \cite[]{Rudin1976,Rockafellar1970} then ensure the existence of a unique continuous extension of $v_s[{n}_s]$ at $n\in \mathscr{B}$, thus the existence of the limit in eq.~(\ref{eq:lin_vn}), which is sufficient to conclude.
\end{proof}

Of course, the hypothesis that a derivative of $v_s[n_s]$ exists in a restricted sense for $n_s\in\mathscr{B}_s$
is a very strong one, stronger than the conjecture that a derivative of $F_L[n]$ exists in a restricted sense for $n\in\mathscr{B}$.
In the following, we give some static linear response-based elements
that hopefully contribute to the reflection on the existence of some derivative of $v_s[n_s]$.
We underline these elements are not sufficiently rigorous from a mathematician's point of view.
However, remind this document must be considered as notes (not an article) whose aim is to gather some possible tracks.

\subsection{{Static linear response in the non-degenerate case}}
\label{sec:diff_2}

We consider a potential $v_s\in \mathscr{V}^{nondeg}\subset\mathscr{V}$
that leads to non-degenerate solutions for the non-interacting system,
i.e. to $n_s[v_s]\in\mathscr{A}_s$.
Note that sufficiently small perturbations $\delta w\in\mathscr{V}$ of the corresponding potential should also lead to non-degenerate solutions,
i.e. to $n_s[v_s+\delta w]\in\mathscr{A}_s$~\cite[]{vanLeeuwen2003}.

%
The static linear response applied to the non-interacting non-degenerate case gives the following relationship between a perturbation $\delta w$ that represents a first-order change in $v_s$ and the corresponding first-order change $\delta m$ in $n_s$:
\begin{eqnarray}
\label{eq:diff_1}
\delta m[v_s,\delta w](r)
&=&
\int_{\mathbb{R}^3} \chi_s[v_s](r',r) \delta w(r') dr'
,
\end{eqnarray}
where $\chi_s[v_s]$ represents the ``static density response function'' defined in \ref{app:proofs_1_1}.
If (as usual) we identify perturbations of the potential defined up to an additive constant,
we can deduce from eq.~(\ref{eq:diff_1}) that $\delta m$ is a unique functional of $v_s$ and $\delta w$ \cite[]{vanLeeuwen2003}.
The obtained $\delta m[v_s,\delta w]$ satisfies $\int_{\mathbb{R}^3}\delta m[v_s,\delta w](r)dr=0$, see \ref{app:proofs_1_1} for details, which is an expected physical property as we do not want the density perturbations to change the particles number $N=\int_{\mathbb{R}^3}n_s(r)dr$.
$\delta m[v_s,\delta w]$ belongs to the set $\delta\mathscr{A}_s^{(n_s)}$ of first-order density changes productible by first-order potential changes in the non-degenerate case.
Note that the set $\delta\mathscr{A}_s^{(n_s)}$ is not well qualified and \textit{a priori} depends on $n_s$\footnote{
We should have $\delta\mathscr{A}_s^{(n_s)}\subset\delta\mathscr{S}$, where $\delta\mathscr{S}$ is
the Banach space that has been introduced in eq.~(\ref{eq:S_0}).
}.

We now would like to write $n_s[v_s+\epsilon\delta w]$
in function of $n_s[v_s]+\delta m[v_s,\epsilon\delta w]$,
where the scaling factor $\epsilon> 0$ is choosen small enough to produce a first-order perturbation $\epsilon\delta w$ of the potential $v_s$.
By construction, $n_s[v_s]+\delta m[v_s,\epsilon\delta w]$ approximates $n_s[v_s+\epsilon\delta w]\in \mathscr{A}_s$ to first-order in $\epsilon$.
A first-order only approximation difficulty is that $n_s[v_s]+\delta m[v_s,\epsilon\delta w]\notin \mathscr{A}_s$, i.e. does not in general exactly represent a non-degenerate non-interacting ground state density for the potential $v_s+\epsilon\delta w$.

Following \cite{vanLeeuwen2003} (sections 8 and 9),
we introduce a ``corrective'' higher order term $\Delta m^{(v_s,\epsilon\delta w)}$
that leads to $n_s[v_s+\epsilon\delta w]=n_s[v_s]+\delta m[v_s,\epsilon\delta w] + \Delta m^{(v_s,\epsilon\delta w)}\in \mathscr{A}_s$ and should satisfy $\lim_{\epsilon\rightarrow 0^+}\Delta m^{(v_s,\epsilon\delta w)}/\epsilon=0$~\footnote[1]{
We should have $\Delta m^{(v_s,\epsilon\delta w)}\in\delta\mathscr{S}$.
}.
Using eq.~(\ref{eq:diff_1}), we obtain
\begin{eqnarray}
\label{eq:diff_2}
\frac{n_s[v_s+\epsilon\delta w](r)-n_s[v_s](r)}{\epsilon}
&=&
\int_{\mathbb{R}^3} \chi_s[v_s](r,r') \delta w(r') dr' + \frac{\Delta m^{(v_s,\epsilon\delta w)}(r)}{\epsilon}
.
\end{eqnarray}
Taking the limit $\epsilon\rightarrow 0^+$ like in eqs.~(\ref{eq:diff_def_2}) and (\ref{eq:diff_def}), we may deduce
\begin{eqnarray}
\label{eq:diff_3}
\forall v_s\in\mathscr{V}^{nondeg}:
\quad
\chi_s[v_s](r',r)
=
\frac{\partial n_s[v_s](r)}{\partial v_s(r')}
.
\end{eqnarray}
%
%
However, rigorous definitions of the set $\delta\mathscr{A}_s^{(n_s)}$ and of the properties of $\Delta m^{(v_s,\epsilon\delta w)}$ would be necessary to conclude on the existence of such a derivative.

\subsection{{Inverse static linear response in the non-degenerate case}}
\label{sec:diff_3}

\noindent
We now would like to consider the ``inverse'' of eq.~(\ref{eq:diff_1}).
%
%
%
Following \cite{vanLeeuwen2003} (sections 8 and 9) and identifying perturbations of the potential defined up to an additive constant,
we can consider that eq.~(\ref{eq:diff_1}) is invertible and have
\begin{eqnarray}
\label{eq:diff_3}
\delta w[n_s,\delta m](r)
=
\int_{\mathbb{R}^3} \chi_s^{-1}[v_s[n_s]](r,r') \delta m(r') dr'
,
\end{eqnarray}
where $\chi_s^{-1}[v_s[n_s]]$ is defined through
\begin{eqnarray}
\label{eq:diff_3b}
\int_{\mathbb{R}^3} \chi_s^{-1}[v_s[n_s]](r,r')\chi_s[v_s[n_s]](r',r'')dr'=\delta(r-r'')
.
\end{eqnarray}
%
An important point is that eq.~(\ref{eq:diff_3}) is valid only for perturbations $\delta m$ that belong to the set of admissible first-order density perturbations $\delta\mathscr{A}_s^{(n_s)}$.
The fact the set $\delta\mathscr{A}_s^{(n_s)}$ is not well qualified and \textit{a priori} depends on $n_s$ thus represents a difficulty.


Again, $n_s+\epsilon\delta m$ does not necessarilly belong to $\mathscr{A}_s$ for all $\epsilon\delta m\in\delta\mathscr{A}_s^{(n_s)}$ and $\epsilon> 0$.
We introduce the ``smallest correction'' $\Delta m^{(n_s,\epsilon\delta m)}$ so that
$n_s+\epsilon\delta m+\Delta m^{(n_s,\epsilon\delta m)}\in \mathscr{A}_s$~\footnote[2]{
We should have $\Delta m^{(n_s,\epsilon\delta m)}\in\delta\mathscr{S}$.
}.
This correction term should satisfy $\lim_{\epsilon\rightarrow 0^+}\Delta m^{(n_s,\epsilon\delta m)}) / \epsilon=0$~\footnote[1]{
This would in general not be true if $\delta m\in\delta\mathscr{S}$,
as $\Delta m^{(n_s,\epsilon\delta m)}$ could then have to compensate for possible $n_s(r)+\epsilon\delta m(r)<0$ at some positions $r$, even at the limit $\epsilon\rightarrow 0^+$.
}.
We then should have $v_s\big[n_s+\epsilon\delta m+\Delta m^{(n_s,\epsilon\delta m)}\big]=v_s\big[n_s\big]+\delta w\big[n_s,\epsilon\delta m\big]+ o(\epsilon)$,
where $o(\epsilon)$ corresponds to the Landau notation~\cite[]{Rockafellar1970}.
Using eq.~(\ref{eq:diff_3}), we deduce
\begin{eqnarray}
&&
\frac{v_s\big[n_s+\epsilon\delta m+\Delta m^{(n_s,\epsilon\delta m)}\big](r)-v_s\big[n_s\big](r)}{\epsilon}
=
\int_{\mathbb{R}^3} \chi_s^{-1}[v_s[n_s]](r,r') \delta m(r') dr'
+ \frac{o(\epsilon)}{\epsilon}
.
\label{eq:diff_4}
\end{eqnarray}
We obtain a form quite similar to eqs.~(\ref{eq:diff_def_2}) and (\ref{eq:diff_def}).
However, a main difference is that a correction term $\Delta m^{(n_s,\epsilon\delta m)}$
appears inside the first $v_s$ in eq. (\ref{eq:diff_4}).
%
Also, rigorous definitions of the set $\delta\mathscr{A}_s^{(n_s)}$ and of the properties of $\Delta m^{(n_s,\epsilon\delta m)}$ would be necessary before concluding $\chi_s^{-1}[v_s[n_s]](r,r')$ would equal a derivative of the KS potential in a restricted sense (for ``perturbation directions" constrained as 
$
h^{(n_s,\epsilon)}=
\delta m
+
\frac{1}{\epsilon}\Delta m^{(n_s,\epsilon\delta m)}
\text{ with }
\epsilon \rightarrow 0^+
\text{ and }
\epsilon\delta m\in\delta\mathscr{A}_s^{(n_s)}
$
\dots).

%
%

\subsection{{Inverse static linear response in the possibly degenerate case}}
\label{sec:diff_4}

%
We finally consider if previous section's considerations
can be extended to all $n_s\in\mathscr{B}_s$.
%
%
If an equation like eq.~(\ref{eq:diff_1}) can be established also in the degenerate case, i.e. for $n_s\in\mathscr{B}_s\backslash\mathscr{A}_s$, then the considerations of \S\ref{sec:diff_3} would straightforwardly generalize to the degenerate non-interacting system case.
However, instead of eq.~(\ref{eq:diff_1}),
we have in the degenerate case the following relationship between a sufficiently small perturbation $\delta w\in \mathscr{V}$ that represents a first-order change in $v_s$
and the corresponding first-order change $\delta m$ in $n_s$:
\begin{eqnarray}
&&
\delta m[v_s,\delta w](r)
=
\int_{\mathbb{R}^3} \chi_s[v_s](r',r) \delta w(r') dr'
+
\int_{\mathbb{R}^3}\int_{\mathbb{R}^3} \xi_s[v_s,\epsilon\delta w](r,r'',r') \delta w(r'')\delta w(r') dr'' dr'
\nonumber\\
&&
\chi_s[v_s](r',r)
=
\sum_{k=1}^{q_s}
\lambda_k \times
\chi_{s_k}[v_s](r',r)
\nonumber\\
&&
\xi_s[v_s,\delta w](r'',r,r')
=
\sum_{k=1}^{q_s}
\lambda_k \times
\xi_{s_k}[v_s,\delta w](r'',r,r')
\nonumber\\
&&
\delta m[v_s,\delta w](r)
=
\sum_{k=1}^{q_s} \lambda_k \times \delta m_k[v_s,\delta w](r)
.
\label{eq:diff_degen_1}
\end{eqnarray}
%
All terms that appear in eq.~(\ref{eq:diff_degen_1}) are detailed in \ref{app:proofs_1_2_1}.
The important points are, using similar notations than the ones in eq.~(\ref{eq:KS}), 
that the $\lambda_k$ denote the weights associated to each degenerate ground states $|\psi_{s_k}\rangle$ of $\hat{H}_v$ chosen to be orthonormal eigenstates (for instance Slater determinants in the case of fermions), and that ${q_s}$ denotes the degeneracy. 
Each $\delta m_k[v_s,\delta w]$ represents a first-order change in the $k^{th}$ eigenstate density $n_{s_k}(r)=\langle\psi_{s_k}|\hat{n}(r)|\psi_{s_k}\rangle$, so that
$\delta m[v_s,\epsilon\delta w]$ represents a first-order change in the total density $n(r)=\sum_{k=1}^{q_s} \lambda_k n_k(r)$.

One difficulty with the degenerate case, eq.~(\ref{eq:diff_degen_1}), is that
$\delta m[v_s,\delta w]$ does not depend linearly on the perturbation $\delta w$,
unlike in the non-degenerate case, eq.~(\ref{eq:diff_1}).
A second difficulty is that $\xi_s[v_s,\delta w]$ depends on $\delta w$ as explained in \ref{app:proofs_1_2_1}.
A third difficulty is that an explicit form for the inverse of eq. (\ref{eq:diff_degen_1}) is missing.
All this prevents to straightforwardly generalize the considerations of \S\ref{sec:diff_2} and \ref{sec:diff_3} to the degenerate case.

To overcome these difficulties, 
a new result is demonstrated in \ref{app:proofs_1_2_2}: there exists a particular choice for the $\lambda_k$
that always nullifies the $\xi_s[v_s,\delta w]$ term in eq.~(\ref{eq:diff_degen_1}).
This particular choice is $\lambda_k=1/q_s$, i.e. considering the degenerate densities $n_s$ in which 
the degenerate orthonormal eigenstate $|\psi_{s_k}\rangle$ are equally represented.
Then, an equation of the form of eq.~(\ref{eq:diff_3}) is recovered and we can generalize the considerations of \S\ref{sec:diff_3} to the degenerate case
(identifying perturbations of the potential defined up to an additive constant).
Note that the resulting $\delta m[v_s,\delta w]$ satisfies $\int_{\mathbb{R}^3}\delta m[v_s,\delta w](r)dr=0$, 
see \ref{app:proofs_1_2_2}, which is physical.
The corresponding $\delta m[v_s,\delta w]$ belongs to the set $\delta\mathscr{B}_s^{(n_s)}$ of first-order density changes productible by first-order potential changes in this specific configuration.

Does this particular $\lambda_k=1/q_s$ choice make sense?
As explained in \S\ref{sec:Non-int}, degenerate non-interacting systems considerations require to introduce the equivalence class of the degenerate non-interacting ground state densities.
As all elements of this class are considered equivalent, it is sufficient to tag one of them as representative of the of the equivalence class (here the element related to $\lambda_k=1/q_s$).

So, considering this equivalence class and following the steps of \S\ref{sec:diff_3},
we can obtain an equation like eq.~(\ref{eq:diff_2}) with $\delta m[v_s,\delta w]\in\delta\mathscr{B}_s^{(n_s)}$, quite similar to eqs.~(\ref{eq:diff_def_2}) and (\ref{eq:diff_def}).
However, again, a main difference is that a correction term $\Delta m^{(n_s,\epsilon\delta m)}$
appears inside the first $v_s$ in eq. (\ref{eq:diff_4}).
Also, rigorous definitions of the set $\delta\mathscr{B}_s^{(n_s)}$ and of the properties of $\Delta m^{(n_s,\epsilon\delta m)}$ would be necessary before concluding $\chi_s^{-1}[v_s[n_s]](r,r')$ would equal a derivative of the KS potential in a restricted sense (for ``perturbation directions" constrained as 
$
h^{(n_s,\epsilon)}=
\delta m
+
\frac{1}{\epsilon}\Delta m^{(n_s,\epsilon\delta m)}
\text{ with }
\epsilon \rightarrow 0^+
\text{ and }
\epsilon\delta m\in\delta\mathscr{B}_s^{(n_s)}
$
\dots).
However, note that if a more rigourous mathematical analysis of the elements exposed here was one day obtained 
(of course together with a proof that a derivative of $F_L[n]$ can be defined in some restricted sense for all $n\in\mathscr{B}$),
the obtained derivative of the KS potential would represent 
the ``derivative of the equivalence class''.

\subsection{{Unboundedness of the inverse kernel}}
\label{sec:diff_bound}

\ref{app:proofs_1_2_3} shows that the kernel $-\chi_s[v_s[n_s]]$ is real and symmetric,
and has real positive eigenvalues $\alpha_j$ such that $\exists \epsilon>0: \alpha_j > \epsilon$.
The kernel $-\chi_s^{-1}[v_s[n_s]]$ has real eigenvalues $1/\alpha_j$
However, as $\alpha_j$ can become arbitrailly close to $0$, $1/\alpha_j$ is not bounded.
This implies that the representation of $-\chi_s^{-1}[v_s[n_s]]$ in the orthonormal eigenvector basis is not bounded, so as its representations in any other basis.
%
So, even if it was possible to rigorously demonstrate that $\chi_s^{-1}[v_s[n_s]]$ represents a derivative of the KS potential in a restricted sense,
the unboundedness of $\chi_s^{-1}[v_s[n_s]]$ would not allow to conclude that the
non-interacting $v$-representability conjecture is true (remind \S\ref{sec:Non-int-conject}).
Additional constraints would certainly have to enter into account.

%

\section{Conclusion}
\label{sec:concl}

We considered Lieb functional-based DFT.
We reminded that an important open question is related to the existence of a derivative of $F_L[n]$ in a restricted sense for $n\in\mathscr{B}$.
Conjecturing that such a derivative can be defined in some restricted (or weak) sense,
we further investigated the validity of the non-interacting $v$-representability conjecture.
We explained how the proof of the latter conjecture can be reduced
to proving that the KS potential $v_s[n_s]$ admits a derivative in a restricted sense for all $n_s\in\mathscr{B}_s$
and that this derivative is bounded.
However, proving that such a derivative exists in a restricted sense
should be difficult (additionally to the difficulty of proving that a derivative of the Lieb functional exists in a restricted sense).
We gave some static linear response-based elements
that may contribute to the reflection,
underlining points that are not sufficiently rigorous from a mathematician's point of view.
We finally discussed a subtlety occuring in the degenerate non-interacting system case,
related to the importance of a specific equivalence class for degenerate ground state densities.

\newpage
\clearpage

\appendix
\section*{Appendix: Static linear response for a non-interacting system}
\label{app:proofs_1}

\section{\textbf{Non-degenerate case reminders}}
\label{app:proofs_1_1}

We use the notations of \S\ref{sec:diff_2}, considering a non-degenerate non-interacting system described by the Hamiltonian $\hat{H}_s$, eq.~(\ref{eq:H_s}).
The static linear response of the corresponding Schr\"odinger equation gives the relationship in eq.~(\ref{eq:diff_1}),
as detailed in section 7 of \cite{vanLeeuwen2003}.
The ``static density response function'' is defined by
(``c.c.'' denotes the complex conjugate)
\begin{eqnarray}
\label{app:B_s}
\chi_s[v_s](r',r)
=
-\sum_{i=2}^{\infty}
\frac{\langle\psi_{s}^{(1)}|\hat{n}(r')|\psi_{s}^{(i)}\rangle\langle\psi_{s}^{(i)}|\hat{n}(r)|\psi_{s}^{(1)}\rangle}
{E_{s_i}-E_{s_1}}
+c.c.
.
\end{eqnarray}
$\{|\psi_{s}^{(i)}\rangle, i\ge 1\}$ denotes the set of orthonormal eigenstates of $\hat{H}_s$
(for instance Slater determinants in the case of fermions),
which are obviously functionals of the potential $v_s$ (implicit in the following).
The ground state is $|\psi_{s}^{(1)}\rangle$, which equals $|\psi_{s_1}\rangle$ in the notations of \S\ref{sec:reminders}.
The ground state density is $n(r)=\langle\psi_{s}^{(1)}|\hat{n}(r)|\psi_{s}^{(1)}\rangle$.
$E_{s_i}=\langle\psi_{s}^{(i)}|\hat{H}_s|\psi_{s}^{(i)}\rangle$ denotes the energy associated to the eigenstate $i$, $E_{s_1}$ being the ground state energy.

Because $\langle\psi_{s}^{(k)}|\psi_{s}^{(l)}\rangle=\delta_{kl}$, we deduce from eq.~(\ref{app:B_s}) that
$\int_{\mathbb{R}^3} \chi_s[v_s](r',r) dr=0$.
Thus, the $\delta m[v_s,\delta w](r)$ term in eq.~(\ref{eq:diff_1})
satisfies $\int_{\mathbb{R}^3}\delta m[v_s,\delta w](r)dr=0$.

\section{Degenerate case}
\label{app:proofs_1_2}

\subsection{{Reminders}}
\label{app:proofs_1_2_1}

We consider a degenerate non-interacting system described by the Hamiltonian $\hat{H}_s$, eq.~(\ref{eq:H_s}).
$\{|\psi_{s}^{(i)}\rangle, i\ge 1\}$ still denotes the set of orthonormal eigenstates of $\hat{H}_s$
(for instance Slater determinants in the case of fermions).
However, the first eigenstates are now degenerate for $i=1,...,q_s$.
Using the letter $k$ to index these eigenstates (for coherency with the notations of \S\ref{sec:Non-int}),
we have $\forall k\in\{1,...,q_s\}:\langle\psi_{s}^{(k)}|\hat{H}_s|\psi_{s}^{(k)}\rangle=E_{s_1}$.
The density associated to the eigenstate $k$ is $n_{s_k}(r)=\langle\psi_{s}^{(k)}|\hat{n}(r)|\psi_{s}^{(k)}\rangle$.
As explained in \S\ref{sec:Non-int},
there is in the degenerate case a one-to-one correspondence between $v_s$ and the class of ground state densities
\begin{eqnarray}
\label{app:class_densit}
\Big\{
n_s
\Big|
n_s=\sum_{k=1}^{q_s} \lambda_k \times n_{s_k}(r),\lambda_k^*=\lambda_k\ge 0,\sum_{k=1}^{q_s} \lambda_k=1
\Big\}
.
\end{eqnarray}
As detailed in section 11 of \cite{vanLeeuwen2003},
the perturbation theory applied to the corresponding Schr\"odinger equation gives for $k\in\{1,...,{q_s}\}$:
\begin{eqnarray}
&&
\delta m_k[v_s,\delta w](r)
=
\int_{\mathbb{R}^3} \chi_{s_k}[v_s](r',r) \delta w(r') dr'
+
\int_{\mathbb{R}^3}\int_{\mathbb{R}^3} \xi_{s_k}[v_s,\delta w](r,r'',r') \delta w(r'')\delta w(r') dr'' dr'
,
\nonumber\\
\label{app:degen_1}
\end{eqnarray}
where
\begin{eqnarray}
\label{app:degen_2}
&&
\chi_{s_k}[v_s](r',r)
=
-\sum_{i={q_s}+1}^{\infty}
\frac{\langle\psi_{s}^{(k)}|\hat{n}(r')|\psi_{s}^{(i)}\rangle\langle\psi_{s}^{(i)}|\hat{n}(r)|\psi_{s}^{(k)}\rangle}
{E_{s_i}-E_{s_1}}
+c.c.
\\
&&
\xi_{s_k}[v_s,\delta w](r,r'',r')
=
\stackrel[l\ne k]{}{\sum_{l=1}^{q_s}}\sum_{i=q+1}^{\infty}
\frac{
\langle\psi_{s}^{(k)}|\hat{n}(r)|\psi_{s}^{(l)}\rangle
\langle\psi_{s}^{(l)}|\hat{n}(r'')|\psi_{s}^{(i)}\rangle
\langle\psi_{s}^{(i)}|\hat{n}(r')|\psi_{s}^{(k)}\rangle
}{(E'_{s_k}-E'_{s_l})(E_{s_i}-E_{s_1})}
+c.c.
.
\nonumber
\end{eqnarray}
Considering $\forall k\in\{1,...,{q_s}\}: E_{s_k}(\epsilon\delta w)=\langle\psi_{s}^{(k)}[v_s+\epsilon\delta w]|\hat{H}_s+\epsilon\delta w|\psi_{s}^{(k)}[v_s+\epsilon\delta w]\rangle$,
which satisfies
$E_{s_k}=\lim_{\epsilon\rightarrow 0}E_{s_k}(\epsilon\delta w)$,
$E'_{s_k}$ is defined by
$E'_{s_k}=\lim_{\epsilon\rightarrow 0}\partial E_{s_k}(\epsilon\delta w)/\partial\epsilon$.
$E_{s_k}$ is independent on the choice of $\delta w$
(the ground state energy remains the same whatever the path used to approach it),
i.e. is a functional of $v_s$ only.
$E'_{s_k}$ is however dependent on the choice of $\delta w$
(the variation of the ground state energy can depend on the path used),
i.e. is a functional of $v_s$ and $\delta w$.
Thus, $\xi_{s_k}$ is also a functional of $v_s$ and $\delta w$.

Because $\langle\psi_{s}^{(k)}|\psi_{s}^{(l)}\rangle=\delta_{kl}$, we deduce from eq.~(\ref{app:degen_2}) that $\int_{\mathbb{R}^3} \chi_{s_k}[v_s](r',r) dr=\int_{\mathbb{R}^3} \xi_{s_k}[v_s](r,r'',r') dr=0$.
As a consequence, the $\delta m_k[v_s,\delta w]$ in eq.~(\ref{app:degen_1}) satisfies $\int_{\mathbb{R}^3}\delta m_k[v_s,\delta w](r)dr=0$.
A stronger property that can be demonstrated is that each of the two terms that compose $\delta m_k[v_s,\delta w]$ in eq.~(\ref{app:degen_1}) separately integrate to $0$.


From eq.~(\ref{app:degen_1}), we deduce that any degenerate ground state density first-order perturbation $\delta m[v_s,\delta w](r)=\sum_{k=1}^{q_s} \lambda_k \times \delta m_k[v_s,\delta w](r)$ satisfies eq.~(\ref{eq:diff_degen_1}).

\subsection{{Recovering a linear dependency}}
\label{app:proofs_1_2_2}


We demonstrate that there exists a particular choice for the $\lambda_k$
that always nullifies the $\xi_s[v_s]$ term in eq. (\ref{eq:diff_degen_1})
or equivalently the $\xi_{s_k}[v_s]$ terms in eq. (\ref{app:degen_1}).
Defining
\begin{eqnarray}
\label{app:degen_4}
&&
N_{kl}(r)=\langle\psi_{s}^{(k)}|\hat{n}(r)|\psi_{s}^{(l)}\rangle
\quad \Rightarrow\quad
N_{lk}(r)=N_{kl}^*(r)
\\
&&
W_{li}=\int_{\mathbb{R}^3}\langle\psi_{s}^{(l)}|\hat{n}(r'')|\psi_{s}^{(i)}\rangle \delta w(r'')dr''
\quad \Rightarrow\quad
W_{il}=W_{li}^*
,
\nonumber
\end{eqnarray}
we compute
\begin{eqnarray}
\label{app:degen_5}
&&
\int_{\mathbb{R}^3}\int_{\mathbb{R}^3} \xi_s[v_s,\delta w](r,r'',r') \delta w(r'')\delta w(r') dr'' dr'
\\
&&\hspace{1.5cm}
=
\sum_{k=1}^{q_s}
\lambda_k
\stackrel[l\ne k]{}{\sum_{l=1}^{q_s}}\sum_{i={q_s}+1}^{\infty}
\frac{
N_{kl}(r)
W_{li}
W_{ik}
}{(E'_{s_k}-E'_{s_l})(E_{s_i}-E_{s_1})}
+c.c.
\nonumber\\
&&\hspace{1.5cm}
=
\sum_{i={q_s}+1}^{\infty}
\frac{1}{E_{s_i}-E_{s_1}}
\sum_{k=1}^{q_s}
\stackrel[l\ne k]{}{\sum_{l=1}^{q_s}}
\lambda_k
\frac{
N_{kl}(r)
W_{li}
W_{ik}
}{E'_{s_k}-E'_{s_l}}
+c.c.
\nonumber\\
&&\hspace{1.5cm}
=
\sum_{i={q_s}+1}^{\infty}
\frac{1}{E_{s_i}-E_{s_1}}
\sum_{k=1}^{q_s}
\sum_{l>k}^{q_s}
\Big(
\lambda_k
\frac{
N_{kl}(r)
W_{li}
W_{ik}
}{E'_{s_k}-E'_{s_l}}
+
\lambda_l
\frac{
N_{lk}(r)
W_{ki}
W_{il}
}{E'_{s_l}-E'_{s_k}}
\Big)
+c.c.
\nonumber\\
&&\hspace{1.5cm}
=
\sum_{i={q_s}+1}^{\infty}
\frac{1}{E_{s_i}-E_{s_1}}
\sum_{k=1}^{q_s}
\sum_{l>k}^{q_s}
(\lambda_k-\lambda_l)
\frac{
N_{kl}(r)
W_{li}
W_{ik}
}{E'_{s_k}-E'_{s_l}}
+c.c.
.
\nonumber
\end{eqnarray}
If we do the particular choice of equal $\lambda_k$'s,
which implies $\lambda_k=1/{q_s}$, eq. (\ref{app:degen_5}) always equals to zero.
This choice
allows to neglect the second term in eq. (\ref{app:degen_1}) and thus in eq. (\ref{eq:diff_degen_1}).
Because of aforementioned properties, the resulting $\delta m[v_s,\delta w]$ would still satisfy $\int_{\mathbb{R}^3}\delta m_k[v_s,\delta w](r)dr=0$, which makes sense.


\subsection{Kernel invertibility}
\label{app:proofs_1_2_3}

$-\chi_s[v_s[n_s]]$ represents a real and symmetric kernel.
It has real eigenvalues $\alpha_j$ related to orthonormal eigenvectors $f_j(r)$, satisfying
(eq. (\ref{app:degen_2}) is used to deduce the second line)
\begin{eqnarray}
\label{app:bound_2}
&&
-\int_{\mathbb{R}^3}\chi_s[v_s[n_s]](r,r')f_j(r')dr'
=
\alpha_j f_j(r)
\\
&&
\alpha_j
=
2
\sum_{k=1}^{q_s}
\sum_{i=q_s+1}^{\infty}\lambda_k
\frac{\big|\int_{\mathbb{R}^3}\langle\psi_{s}^{(k)}|\hat{n}(r)|\psi_{s}^{(i)}\rangle f_j(r)dr\big|^2}
{E_{s_i}-E_{s_1}}
.
\nonumber
\end{eqnarray}
We must have $\alpha_j > 0$ because a null $\alpha_j$ would imply a constant $f_j(r)$, which is not possible because of the orthonormalization constraint on the $f_j(r)$.
We recover that $\chi_s[v_s[n_s]]$ is invertible.
Using eq.~(\ref{eq:diff_3b}), we have
\begin{eqnarray}
\label{app:bound_3}
-\int_{\mathbb{R}^3}\chi_s^{-1}[v_s[n_s]](r',r)f_j(r)dr
=
\frac{1}{\alpha_j} f_j(r')
,
\end{eqnarray}
i.e. the real and symmetric kernel $-\chi_s^{-1}[v_s[n_s]]$ has real positive eigenvalues $1/\alpha_j$ related to the orthonormal eigenvectors $f_j(r)$.
Note that the $1/\alpha_j$ are not bounded because the $\alpha_j$ can become arbitrarily close to $0$~\cite[]{Par89,Dre90,Engel2011}.


 \newpage
 \clearpage
\bibliographystyle{apalike}

\end{document}